\documentclass[11pt]{article}

\usepackage{microtype}

\usepackage{cite}

\usepackage{amssymb,amsmath,amsthm,mathtools}
\usepackage{xspace}
\usepackage{enumerate}
\usepackage{graphicx}

\usepackage[medium,compact]{titlesec}

\usepackage{fullpage}

\newtheorem{Definition}{Definition}[section]
\newtheorem{Theorem}[Definition]{Theorem}

\newtheorem{Lemma}[Definition]{Lemma}
\newtheorem{Corollary}[Definition]{Corollary}
\theoremstyle{remark}

\usepackage[usenames,dvipsnames]{color}

\newcommand{\R}{\mathbb{R}}

\renewcommand{\O}{\mathcal{O}}
\newcommand{\Sum}{\displaystyle\sum}

\newcommand{\C}{\mathcal{C}}
\renewcommand{\c}{C}

\newcommand{\TT}{{\mathbf{T}}}
\newcommand{\TTf}{{\TT_f}}
\newcommand{\TTg}{{\TT_g}}
\newcommand{\overlay}{\mathcal{C}}

\renewcommand{\triangle}{\mathcal{T}}
\newcommand{\Sv}{\Sigma_v}
\newcommand{\Se}{\Sigma_e}

\renewcommand{\r}{r}

\renewcommand{\Xi}{(X_1,...,X_{\r-1})}

\newcommand{\M}{\mathcal{M}}

\title{Computing the Distance between\\Piecewise-Linear Bivariate
  Functions\thanks{%
    Work on this paper was initiated at the International
    INRIA-McGill-Victoria Workshop on Problems in Computational
    Geometry, held at the Bellairs Research Institute of McGill
    University in Barbados, West Indies.}}
\author{Guillaume Moroz\thanks{%
    INRIA Nancy - Grand Est,
    615 rue du Jardin Botanique,
    54600 Villers-l\`es-Nancy,
    France.}
  \and Boris Aronov\thanks{%
    Department of Computer Science and Engineering,
    Polytechnic Institute of NYU, Brooklyn, NY~11201-3840, USA;
    \textsl{aronov@poly.edu}.
    Work by B.A.\ on this paper has been supported by grant No.~2006/194
    from the U.S.-Israel Binational Science Foundation, by NSF
    Grant CCF-08-30691, and by NSA MSP Grant H98230-10-1-0210.}}

\date{}

\begin{document}
    \begin{titlepage}
      \maketitle
      \begin{abstract}
        We consider the problem of computing the distance between two
        piecewise-linear bivariate functions $f$ and $g$ defined over
        a common domain $M$.  We focus on the distance induced by the
        $L_2$-norm, that is $\|f-g\|_2=\sqrt{\iint_M (f-g)^2}$.  If $f$
        is defined by linear interpolation over a triangulation of $M$
        with $n$ triangles, while $g$ is defined over another such
        triangulation, the obvious na\"ive algorithm requires
        $\Theta(n^2)$ arithmetic operations to compute this distance.
        We show that it is possible to compute it in $\O(n\log^4 n)$
        arithmetic operations, by reducing the problem to multi-point
        evaluation of a certain type of polynomials.

        We also present an application to terrain matching.
      \end{abstract}
      \thispagestyle{empty}
    \end{titlepage}

    \section{Introduction and problem statement}

    In this paper we use a novel combination of tools from
    computational geometry and computer algebra to speed up a
    computation of $\|\cdot\|_2$-norm distance between two bivariate
    piecewise-linear functions.  Algebraic tools have already been
    used in other recent work in computational geometry, to seemingly
    defy ``obvious'' lower bounds: for example, Ajwani, Ray, Seidel,
    and Tiwary~\cite{centroid} use algebraic tools to compute the
    centroid of all vertices in an arrangement of $n$ lines in the
    plane, without explicitly computing the vertices.

    We feel that working on computational geometry problems using a
    combination of traditional and algebraic tools expands the
    repertoire of questions that can be approached and answered
    satisfactorily.  Employing such combinations of methods expands
    the horizon of solvable problems.  Indeed, in a significant recent
    development, several breakthrough results have been obtained by
    applying algebraic methods to problems of combinatorial geometry:
    Guth and Katz's recent work on the joints
    problem~\cite{Guth-Katz-joints} and on the Erd\H os distinct
    distance problem~\cite{Guth-Katz-distinct} has triggered an
    avalanche of activity; see, for example,
    \cite{EKS,ES,KMS,Mecg11,Q,Solymos-Tao,KSS-lines-joints}.  It
    appears that the use of algebraic tools in geometry (both
    combinatorial and computational) allows one to approach problems
    inaccessible by more traditional methods.  The present work is
    just one step in that direction.

    \paragraph*{Background and previous work}
    
    In~\cite{terrains}, Aronov et al. considered a quite common
    object in computational geometry and geographical information
    systems (GIS): that of a ``terrain.''  A \emph{terrain} is (the graph of)
    a bivariate function over some planar domain, say, a rectangle or a
    square.  It is often used to model geographic terrains, e.g.,
    elevation in a mountainous locale, but can also be applied to
    storing any two-dimensional data sets, such as precipitation or
    snow cover data.  A common (though by no means the only) way of
    interpolating and representing discrete two-dimensional data is a
    \emph{triangulated irregular network} (\emph{TIN}): the values of
    a bivariate function are given at discrete points in, say, the
    unit square.  The square is triangulated, using data points as
    vertices.  The function is then linearly interpolated over each
    triangle.  This produces a piecewise-linear approximation of the
    ``real'' (and unknown) function.

    The problem raised in~\cite{terrains} was that of comparing two
    terrains over the same domain, say, the unit square, but given over
    two unrelated triangulations.  One could imagine comparing the
    outcome of two different ways of measuring the same data, or
    finding correlation between, say, the elevation and the snow cover
    over the same geographic region.  The focus of that work was on
    identifying linear dependence between the two functions or
    terrains.  Three natural distance measures between the two
    functions were considered and several algorithms presented for
    computing such a distance and optimizing it, subject to vertical
    translation and scaling.  The only observation made for the
    $\|\cdot\|_2$ norm (see the definitions below) in~\cite{terrains}
    is that, if the two terrains share a triangulation, both the
    distance computation and the optimization problem can be solved
    easily in linear time, while it appears that in general quadratic
    time seems to be needed to deal with the case of arbitrarily
    overlapping triangulations.  The substance of the current work is
    disproving this assertion and describing a near-linear-time
    algorithm for both problems.

    \paragraph*{Problem statement and results}

    Given bivariate functions $f,g \colon M \to \R$, one can naturally
    define a distance between them as
    \[
    \|f-g\|_2 = \bigl(\iint_M (f(x,y)-g(x,y))^2 dx dy\bigr)^{1/2}.
    \]
    Expanding the expression under the integral, we obtain $\iint
    (f-g)^2 = \iint f^2 - 2 \iint fg + \int g^2$.  If the two
    functions are piecewise linear, defined over different
    triangulations of $M$, only the middle term presents a problem for
    efficient computation. Thus, in the bulk of the paper we will
    focus on the computation of $\iint fg$, showing the following:

    \begin{Theorem}
      \label{thm:product}
      Given piecewise-linear functions $f$ and $g$ defined over
      different triangulations of the same domain $M$, with $n$
      triangles each, $\iint_M f(x,y)g(x,y) dx dy$ can be computed
      using $\O(n\log^4 n)$ arithmetic operations.      
    \end{Theorem}

    Armed with this result, as already mentioned, we can quickly
    compute $\|f-g\|_2$:

    \begin{Theorem}
      \label{thm:distance}
      Given piecewise-linear functions $f$ and $g$ defined over
      different triangulations of the same domain $M$, with $n$
      triangles each, $\|f-g\|_2$ can be computed using $\O(n\log^4
      n)$ arithmetic operations.
    \end{Theorem}

    Na\"ively, the integral in Theorem~\ref{thm:product} can be
    expressed as a sum of integrals over each cell appearing in the
    overlay of the two triangulations of $f$ and $g$.  Unfortunately,
    this overlay has a quadratic number of cells, in the worst case.
    The main idea of our algorithm is to reduce the computation of the
    integral to double sums of algebraic functions over grids, which allows
    us tu use fast multi-point evaluation algorithms.

    The paper is organized as follows. In section~\ref{sec:triangles},
    we will show how the integral of a function over a convex polygon
    can be expressed as a sum of elementary algebraic functions over
    its vertices.  Applying the process to a convex decomposition of a
    region $M$ expresses the value of $\iint_M fg$ as a summation over
    the vertices of the decomposition.  Then, in
    section~\ref{sec:triangulation}, we show how the integral over the
    overlay of the two triangulations can be reduced to a sum of
    elementary functions over pairs of edges, plus some additional
    terms computable in linear time. In section~\ref{sec:bipartite},
    we use the bipartite clique decomposition to arrange the pairs of
    edges in complete grids of fairly rigid form. Finally, in
    section~\ref{sec:multipoint}, we use a fast multi-point evaluation
    algorithm to compute the sums over each grid, completing the
    description of our method.  

    %


    \section{How to integrate over a convex subdivision}
    \label{sec:triangles}
    
    Consider a bivariate function $h$ defined over a convex polygon $\c$ in
    the plane.  To simplify our presentation and without loss of generality, we
    assume that all vertices of $\c$ lie to the right of the $y$-axis.  For
    a vertex $p$ of $\c$ we refer to the lines supporting the edges of $\c$
    incident to $p$ as $L(p,\c)$ (the one with higher slope) and $U(p,\c)$
    (the one with lower slope); to the left of $p$, $L(p,\c)$ is
    below $U(p,\c)$. Let
    \begin{align*}
      L(p,\c)&\colon y=y_l(p,\c)+s_l(p,\c)x, \text{ and}\\
      U(p,\c)&\colon y=y_u(p,\c)+s_u(p,\c)x.
    \end{align*}
    We omit the explicit dependence on $C$ and/or $p$ whenever it causes no
    confusion.  Define
    \[\delta(p,\c)=\begin{dcases*}  
      -1 & if $\c$ is above both $L(p,\c)$ and $U(p,\c)$, or below both of them,\\
      +1 & if $\c$ is below $U(p,\c)$ and above $U(p,\c)$, or vice versa.
                       \end{dcases*}
    \]
    Finally, put
    \[
    \triangle(p,\c,h) :=
    \delta(p,\c)\int_{x=0}^{x_{p_j}}\int_{y=y_l(p,\c)+s_l(p,\c)x}^{y_u(p,\c)+s_u(p,\c)x}
    h(x,y)dxdy.
    \]
    In words, $\triangle$ is the signed integral of $h$ over the triangle
    $T_p$ delimited by the $y$-axis, $L(p,\c)$, and $U(p,\c)$.

    With the above notation, we express the integral of $h$
    over $\c$ in a convenient way as a sum of terms associated with
    its vertices:
    
    \begin{Lemma}
    \label{lem:sum}
    Let $\c$ be a convex polygon with vertices $p_1,...,p_k$,
    and $h$ a bivariate function.  Then
    \begin{equation}
      \label{eq:sum}
      \iint_{\c} h(x,y)dxdy = \Sum_{j=1}^k 
      \triangle(p_j,\c,h).
      \end{equation}
    \end{Lemma}
    
    
    \begin{proof}
      Partition the vertices of $\c$ into four subsets $V_L$,
      $V_R$, $V_T$, and $V_B$ as follows.
      $V_L:=\{p_L\}$ (resp., $V_R:=\{p_R\}$) consists
      of the unique leftmost (resp., rightmost) vertex of $\c$.
      $V_T:=\{{t}_1,\ldots\}$ (resp.,
      $V_B:=\{{b}_1,\ldots\}$) is the sequence of
      vertices of $\C$ from $p_R$ to $p_L$ in the
      counterclockwise (resp., clockwise) direction; refer to
      Fig.~\ref{fig:cell}.
    \begin{figure}
      \centering
      \includegraphics{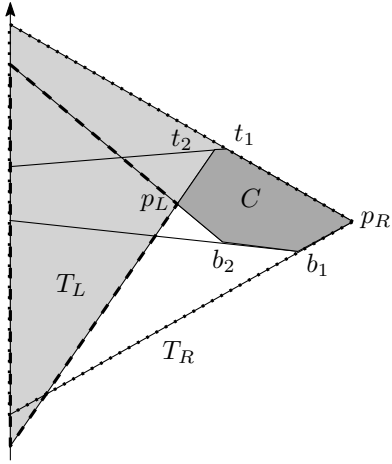}
      \caption{An illustration of the proof of Lemma~\ref{lem:sum}.
        $T_R$ is dotted; $T_L$ is dashed; $A_T$ is lightly
        shaded.}
      \label{fig:cell}
    \end{figure}
    
    To each vertex $p$, we associate the triangle $T_p$ as defined
    above.  Put $T_L:=T_{p_L}$ and $T_R:=T_{p_R}$.  Since $\c$ is
    convex, so for all $i$, line ${t}_i{t}_{i+1}$ (resp.,
    ${b}_i{b}_{i+1})$ is below the line ${t}_{i-1}{t}_i$ (resp., above
    the line ${b}_{i-1}{b}_i$) left of $\c$.
    Thus, the triangles $T_t$, $t \in V_T$, (resp., $T_b$, for
    $b\in V_B$) do not overlap.  Let $A_T = \bigcup_{t\in
      V_T} T_p$ and $A_B = \bigcup_{b\in V_B} T_b$.  By
    construction we have
    \[A_T \cup A_B = \left(T_L \cup T_R \right) \setminus \c
    \quad\text{and}\quad
    A_T \cap A_B = T_L \cap T_R.\]
    Since $\c \subset T_L \cup T_R$, 
    the first equality, written in terms of characteristic functions, gives
    \[1_{A_T} + 1_{A_B} - 1_{A_T\cap A_B} = 1_{T_L} + 1_{T_R} - 1_{T_L\cap T_R} - 1_{\c},\]
    which can be simplified, using the second equality, to yield, as promised
    \[1_{A_T} + 1_{A_B} = 1_{T_L} + 1_{T_R} - 1_{\c}. \qedhere \]
    \end{proof}

    We now consider a convex subdivision $\C$ of some bounded region
    $M$ in the plane, with each cell $\c$ associated with its own
    bivariate function $h_\c$, thereby defining a function~$h$ over
    all of $M$ (as it does not affect the value of the integral, the
    functions $h_C$ need not agree along the common boundaries of
    adjacent cells).  By summing eq.~\eqref{eq:sum} over the cells
    $\c$ of $\C$, we can compute $\iint_M h(x,y) dxdy$:
    
    \begin{Corollary}
    \label{cor:doublesum}
    \begin{equation}
    \label{eq:doublesum}
    \iint_M h(x,y) dxdy =
    \Sum_{\text{cell $\c\in\C$}}\iint_{\c} h_\c(x,y)dxdy =
        \mathop{\Sum\Sum}_{\substack{\text{cell $\c \in\C$}\\
        \text{$\c$ adjacent to $p$}}}
        \triangle(p,\c,h_\c).
         \end{equation}
    \end{Corollary}
 
    
    
    One of the advantages of the formulation in eq.~\eqref{eq:doublesum}
    for our application is that an individual integral under the sum
    can be expressed as a rational function when $h_\c$ is a bivariate
    polynomial.
    
    \begin{Lemma}
      Let $p$ be an intersection point of $y=y_l+s_lx$ and
      $y=y_u+s_ux$ as above.  Then
      \[x_p = - \frac{y_u-y_l}{s_u-s_l}\]
      and
      \[\int_{x=0}^{x_{p}}\int_{y=y_l+s_lx}^{y_u+s_ux}
      x^iy^jdxdy = \frac{P_{i,j}(y_l,y_u,s_l,s_u)}{(s_u-s_l)^{i+j+1}},\]
    where $P_{i,j}$ is a polynomial of total degree $i+2j+2$.
    \label{lem:formula}
    \end{Lemma}
    
    \begin{proof}
    Let $Q_{i,j}(u,v,x)$ be the only polynomial such that
    $\frac{\partial Q_{i,j}}{\partial x}=x^i(u+vx)^j$ and $Q_{i,j}(u,v,0)=0$
    for all $u,v\in\R$.  In particular, $Q_{i,j}$ has total degree
    $i+2j+1$, its degree in $x$ is $i+j+1$, and the coefficient of $x^{i+j+1}$ in $Q_{i,j}$ is $\frac{1}{i+j+1}v^j$.
    Now 
    \begin{align*}
    \int_{x=0}^{x_p}\int_{y=y_l+s_lx}^{y_u+s_ux} x^iy^jdxdy 
                            = & \int_{x=0}^{x_p}\frac{1}{j+1}
                                      (x^i(y_u+s_ux)^{j+1}-x^i(y_l+s_lx)^{j+1})dx\\
                            = & \frac{1}{j+1}(Q_{i,j+1}(y_u,s_u,x_p)
                                              -Q_{i,j+1}(y_l,s_l,x_p)).
    \end{align*}
    Therefore the value of the integral can be expressed as a
    polynomial in $y_l,y_u,s_l,s_u,x_p$ of total degree
    $i+2j+3$ of the form
    \[ \frac{1}{j+1}(s_u^{j+1}-s_l^{j+1})x_p^{i+j+2}
       + \Sum_{k=0}^{i+j+1}q_k(y_u,s_u,y_l,s_l)x_p^k. \]
    After substituting $-\frac{y_u-y_l}{s_u-s_l}$ for $x_p$ and bringing to a
    common denominator,
    we conclude that the expression can be rewritten in the form
    \[\frac{(s_u-s_l)P(y_l,y_u,s_l,s_u)}{(s_u-s_l)^{i+j+2}},\]
    as claimed.
    \end{proof}
    
    \section{How to integrate over an overlay}
    \label{sec:triangulation}
    
    In our problem, we are interested in computing the integral $\iint
    f(x,y)g(x,y)dxdy$, where $f$ is defined over a triangulation
    $\TTf$ of $M \subset \R^2$, with a separate linear function $f_\Delta(x,y)$
    determining $f$ over each triangle $\Delta\in\TTf$;
    $g$ is defined similarly over a different triangulation $\TTg$ of
    $M$.  The product $h(x,y):=f(x,y)g(x,y)$ is thus naturally defined
    over the convex decomposition $\overlay$ of $M$ that is the overlay of
    $\TTf$ and $\TTg$.  By Corollary~\ref{cor:doublesum}, it is
    sufficient to evaluate a sum over all vertices of $\overlay$.  The
    vertices of $\overlay$ come in two flavors: the original vertices of
    $\TTf$ and of $\TTg$, and the intersections of edges of $\TTf$ and
    $\TTg$.  Therefore,
    \begin{align*}
    \iint_M f(x,y)g(x,y)dxdy & {} =
    \mathop{\Sum\Sum}_{\substack{\text{$p$ vertex of $\overlay$}\\
                                  \text{$\c$ adjacent to $p$}}}
           \triangle(p,\c,h_\c)\\
                               & {} = \underbrace{\Sum_{\substack{p \text{ vertex}\\
                                                \text{of $\TTf\cup\TTg$}}}
                                \Sum_{\substack{\c\text{ adjacent}\\
                                                \text{to $p$}}}
           \triangle(p,\c,h_\c)}_{\Sv}
                              + \underbrace{\Sum_{\substack{p=e_1\cap e_2,\\
                                                (e_1,e_2)\in\TTf\times\TTg}}
                                 \Sum_{\substack{\c\text{ adjacent}\\
                                                 \text{to $p$}}}
           \triangle(p,\c,h_\c).}_{\Se}
    \end{align*}
    $\Sv$ involves $\O(n)$ integrals.  We preprocess each of $\TTf$
    and $\TTg$ for logarithmic-time point location queries, in $\O(n
    \log n)$ time (see, for example, \cite{4M}).  For each vertex $p$
    of $\TTf$, we locate the triangle $\Delta_g \in \TTg$ containing
    it.  Then, for each triangle $\Delta_f \in \TTf$ incident to $p$,
    we can compute
    $\triangle(p,\Delta_g\cap\Delta_h,f_{\Delta_f}g_{\Delta_g})$ in
    constant time.  The treatment of vertices of $\TTg$ is symmetric.
    We spend $\O(n \log n)$ total for point location. The remaining
    work is proportional to the sum of vertex degrees in both
    triangulations, which is $\O(n)$.  Hence $\Sv$ can be computed in
    $\O(n \log n)$ arithmetic operations.  We devote the rest of the
    discussion to computing $\Se$ efficiently.

    
    
   \section{Bipartite clique decomposition}
   \label{sec:bipartite}

   Let $E_f$ be the set of edges of $\TTf$ and $E_g$ the set of edges
   of $\TTg$.  In \cite{cegs-ablsp-94}, it is shown that it is
   possible to compute a family
   $\mathcal{F}=\{(R_1,B_1),\ldots,(R_u,B_u)\}$ where $R_k\subset E_f$
   and $B_k\subset E_g$, such that
   \begin{enumerate}[(i)]\itemsep=-\parsep 
   \item
     every segment in $R_k$ intersects every segment in $B_k$;
   \item
     every segment of $R_k$ has lower slope than every segment of
     $B_k$, or vice versa;
   \item for every intersecting pair $(e_1, e_2) \in E_f \times E_g$ there
     is exactly one $k$ such that $e_1 \in R_k$, $e_2 \in
     B_k$; no such $k$ exists for a non-intersecting pair $(e_1,e_2) \in
     E_f \times E_g$;
   \item
     $\sum_{k} (|R_k|+ |B_k|) = O(n\log^2n)$.
   \end{enumerate}
   This family can be computed in $\O(n\log^2n)$ time. 
    
    \section{Multi-point evaluation}
    \label{sec:multipoint}
    
    In this section, we explain how to efficiently compute 
    \begin{equation}
    \Sum_{\substack{p=e_1\cap e_2,\\
                     (e_1,e_2)\in R\times B}}
                                 \Sum_{\substack{\c\text{ adjacent}\\
                                                 \text{to $p$}}}
                                                 \triangle(p,\c,h_\c)
    = \Sum_{e_1\in R} \Sum_{e_2\in B} 
                                 \Sum_{\substack{\c\text{ adjacent}\\
                                                 \text{to $p=e_1\cap e_2$}}}
                                                 \triangle(p,\c,h_\c),
    \label{eq:gridsum}
    \end{equation}
    where $(R,B):=(R_k,B_k)$ is a pair of sets of triangulation edges
    produced by the bipartite clique decomposition.  In particular,
    all segments of $R$ intersect all segments of $B$ and, moreover,
    without loss of generality, the slopes of segments of $R$ are
    greater than those of segments of $B$.
    
    
    \subsection{Reduction to sums of rational functions}
    
    Each triangle of $\TTf$ and $\TTg$ is associated with a bivariate linear
    function.  If $e$ is an edge of $\TTf\cup\TTg$, let $f_u(e)$ be
    the linear function associated to the upper triangle (in the
    triangulation to which $e$ belongs) adjacent to $e$; $f_l(e)$ is
    the corresponding function for the lower triangle; we can define
    the function to be identically zero for regions outside $M$, but
    it will never be used by the algorithm.
    
    A vertex $p=e_1\cap e_2$, with $e_1\in R$ and $e_2\in B$, lies on
    the boundary of four cells of $\overlay$.  We focus on the cell
    $\c_{left}=\c_{left}(p)$ lying above $e_1$ and below $e_2$, for
    which $p$ is the rightmost point.  Thus $h_{\c_{left}}=
    f_u(e_1)f_l(e_2)$ and $\delta(p,\c_{left})=+1$.  Suppose
    $f_u(e_1)\colon (x,y) \mapsto a(e_1)+b(e_1)x+c(e_1)y$ and
    $f_l(e_2)\colon (x,y) \mapsto a(e_2)+b(e_2)x+c(e_2)y$.
    We compute the
    contribution to the sum~\eqref{eq:gridsum} of such cells, over all
    choices of $p$.  (The remaining three types of cells adjacent to
    $p$ are treated by an entirely symmetric argument.)
%
%
%
%
%
    Given an edge $e$ of $\TTf\cup\TTg$, let $y=y(e)+s(e)x$ be the
    equation of the line supporting it, so we can write
    \begin{multline*}
    \triangle(p,\c_{left}, f_u(e_1)f_u(e_2)) = {}\\
           \int_{x=0}^{x_{p}}\int_{y=y(e_1)+s(e_1)x}^{y(e_2)+s(e_2)x}
              \left(\begin{gathered}
              a(e_1)a(e_2) + (a(e_1)b(e_2)+a(e_2)b(e_1))x + {}\\
              (a(e_1)c(e_2)+a(e_2)c(e_1))y + b(e_1)b(e_2)x^2 + {}\\
              (b(e_1)c(e_2)+b(e_2)c(e_1))xy + c(e_1)c(e_2)y^2\end{gathered}\right) dxdy.
    \end{multline*}
    The above integral can be expressed as the sum of nine integrals
    of a function of the form $v(e_1)w(e_2)x^iy^j$, where $v$ and $w$
    are some functions that assign a real number to each edge.
    Gathering all the terms and recalling that $p=e_1\cap e_2$,
    eq.~\eqref{eq:gridsum} can be rewritten as the sum of $36$
    expressions of the form
    \begin{multline}
      \Sum_{e_1\in R} \Sum_{e_2\in B} 
      \int_{x=0}^{x_{e_1\cap e_2}}\int_{y=y(e_1)+s(e_1)x}^{y(e_2)+s(e_2)x}
            v(e_1)w(e_2)x^iy^j\\
       = \Sum_{e_1\in R} \Sum_{e_2\in B}
            \frac{v(e_1)w(e_2)P_{i,j}(y(e_1),y(e_2),s(e_1),s(e_2))}
                 {(s(e_2)-s(e_1))^{i+j+1}},
    \label{eq:fraction}
    \end{multline}
    by Lemma~\ref{lem:formula}.
    
    \subsection{Fast multi-point evaluation}
    
    Now we will use multi-point evaluation to speed up the computation
    of eq.~\eqref{eq:fraction}.  To accomplish this, we will replace
    the values associated to the edges of $R$ by symbolic variables,
    while using the actual numerical values for those for the edges of
    $B$.  Then we will compute the corresponding symbolic rational
    function using a divide-and-conquer strategy (Lemma
    \ref{lem:addfraction}).  Finally, we will use multi-point
    evaluation on the resulting polynomials (Lemma
    \ref{lem:multipoint}).
    
    \begin{Lemma}
    \label{lem:addfraction}
    Let $u$ and $v$ be two functions from the edges of $B$ to $\R$
    and suppose $|B|\leq n$.
    Then
    \begin{equation}
      \Sum_{e\in B} \frac{u(e)}{(X-v(e))^d}
      \label{eq:frac}
    \end{equation}
    can be expressed in the form
    \[\frac{N(X)}{D(X)},\]
    where $N(X)$ and $D(X)$ are polynomials of degree at most $(n-1)d$
    and $nd$ respectively; their coefficients can be computed
    explicitly in $\O(\M(nd)\log n )$ arithmetic operations, where
    $\M(q)=\O(q \log q)$ is the cost of multiplication of two
    univariate polynomials of degree at most $q$.
    \end{Lemma}
    
    \begin{proof}
      For simplicity of presentation and without loss of generality,
      assume that $|B|$ is a power of two.  We bring eq.~\eqref{eq:frac}
      to a common denominator by combining the fractions in pairs,
      reducing their number to $|B|/2$, and repeating the process
      $\log |B|=\O(\log n)$ times.  The bounds on the degree of
      the final numerator and denominator are immediate from examining
      the original fractions.

      We now explain how to bring
      \[\frac{N_1(X)}{D_1(X)} + \frac{N_2(X)}{D_2(X)},\]
      with $N_1,N_2,D_1,D_2$ of degree at most $kd$, to a
      common denominator in time $3 \M(kd)+\O(kd)$.  Indeed, the above
      fraction is equal to
      \[\frac{N_1(X)D_2(X) + N_2(X)D_1(X)}{D_1(X)D_2(X)},\]
      so it can be computed by three calls to fast polynomial
      multiplication plus a linear number of additional operations, as
      claimed.

      This completes the proof of the lemma, as the cost of one round
      of combining fractions with denominators and numerators of
      degree at most $kd$ is $n/k \cdot (3 \M(kd)+\O(kd))=\O(\M(nd))$, since
      $M$ is superlinear. 
    \end{proof}
    
    The second lemma handles summing the values of a special kind of
    polynomials.
    
    \begin{Lemma}
    \label{lem:multipoint}
    Let $P(X_0,X_1,\ldots,X_r)$ be a polynomial of degree $n$ in $X_0$
    and $d$ in $X_1,\ldots,X_r$.  Let $E$ be a set of at most $n$
    points of $\R^{r+1}$.   The values of $P$ at the points of $E$
    can be simultaneously computed in $\O(\binom{d+r}{r}\M(n)\log n)$ time.
    \end{Lemma}
    
    \begin{proof}
      $P$ can be expanded with respect to the variables
      $X_1,\ldots,X_r$, and has at most $\binom{d+r}{r}$ monomials.
      Each coefficient is a univariate polynomial in $X_0$ of degree
      at most $n$.  Then, using standard multi-point evaluation
      algorithm for univariate polynomials \cite[ch. 10]{GGbook03}, we can
      compute simultaneously the values of these univariate coefficients
      of $P$ in $\O(\binom{d+r}{r}\M(n)\log n)$ at every point of $E$.
      Finally, we compute the values of each of the monomials of
      $P$ in time $\O(\binom{d+r}{r}n)$.  Combining all these values
      also costs $\O(\binom{d+r}{r}n)$ arithmetic operations,
      concluding the proof.
    \end{proof}
    
    Now we are ready to efficiently evaluate eq.~\eqref{eq:fraction}.
    Let $F(X,Y,V)$ be the polynomial
    \[F(X,Y,V) := \Sum_{e_2\in B} \frac{V
      w(e_2)P_{i,j}(Y,y(e_2),X,s(e_2))} {(s(e_2)-X)^{i+j+1}}.\] 
    After
    expanding the numerator of this fraction, we note that $F$ has the
    form
    \[F(X,Y,V) = \Sum_{\substack{0\leq d_X\leq i+2j+2\\
                                  0\leq d_Y\leq i+2j+2}}
                  \left(
                  \Sum_{e_2\in B} \frac{c_{d_X,d_Y,d_V}(w(e_2),y(e_2),s(e_2))}
                                         {(s(e_2)-X)^{i+j+1}}\right)
                                         X^{d_X}Y^{d_Y}V.\]
    In our case, $i+j+1\leq 3$, and, using Lemma~\ref{lem:addfraction},
    we can compute each coefficient of $X^{d_X}Y^{d_Y}V$ and
    express $F$ in the following form, in $\O(n\log^2 n)$ time:
    \[F(X,Y,V) = \frac{N(X,Y)V}{D(X)},\]
    with $N$ a polynomial of degree $n(i+j+1)+1$ in $X$, $i+j+2$ in
    $Y$, and $D$ a univariate polynomial of degree $n(i+j+1)$.
    
    Finally, eq.~\eqref{eq:fraction} can be rewritten as 
    \[\Sum_{e_1\in R} \Sum_{e_2\in B}
    \frac{v(e_1)w(e_2)P_{i,j}(y(e_1),y(e_2),s(e_1),s(e_2))}
    {(s(e_2)-s(e_1))^{i+j+1}} = \Sum_{e_1\in R}
    \frac{N(s(e_1),y(e_1))v(e_1)}{D(s(e_1))}.\] 
    As $|R|<n$, using Lemma~\ref{lem:multipoint}, we can compute
    simultaneously all the terms under the sum, and add them together in
    $\O(n\log^2n)$ arithmetic operations.
    
    \subsection{Putting it together}
    
    To summarize, we can evaluate eq.~\eqref{eq:fraction} in
    $\O((|B_k|+|R_k|)\log^2 (|B_k|+|R_k|))$ operations, for each pair
    $(R_k,B_k)$.  We also saw in section~\ref{sec:bipartite} that
    $\sum_{k} (|R_k|+|B_k|) = \O(n\log^2n)$, where $n$ is the number
    of triangles appearing in our triangulations.  Therefore
    \[
      \Sum_{k} (|R_k|+|B_k|) \log^2(|R_k|+|B_k|) \leq \Sum_{k}
      (|R_k|+|B_k|) \log^2 n = \O(n \log^4 n),
    \]
    which allows us to conclude that we can compute $\Se$ in
    $\O(n\log^4n)$ time.  This completes the proof of
    Theorem~\ref{thm:product}. 
    
    \section{Discussion}
    \label{sec:discussion}
    
    In~\cite{terrains}, the following optimization problem was
    considered: given two functions $f$ and $g$ and a distance measure
    $\|f-g\|$ between them (the paper discusses $\|\cdot\|_1$,
    $\|\cdot\|_2$, and $\|\cdot\|_\infty$, but we only consider
    $\|\cdot\|_2$ here), find the values of real parameters $s$ and
    $t$ that minimize $\|f-(sg+t)\|$.  If $f$ and $g$ are interpreted
    as geometric ``terrains,'' we are looking for the scaling and
    translation of the vertical coordinate of the terrain $g$ to best
    match terrain $f$~\cite{terrains}.  Since $\|f-(sg+t)\|_2^2=\iint
    (f-(sg+t))^2$ is a degree-two polynomial in $s$ and $t$ with
    coefficients easily expressible in terms of $\iint f$, $\iint g$,
    $\iint f^2$, $\iint g^2$, $\iint fg$, and $\iint 1$, being able to
    compute $\iint fg$ efficiently immediately yields
    \begin{Theorem}
      \label{thm:matching}
      Given piecewise-linear functions $f$ and $g$ defined over
      different triangulations of the same domain $M$, with $n$
      triangles each, the values $s$ and $t$ minimizing $\|f-(sg+t)\|_2$
      can be computed using $\O(n\log^4 n)$ arithmetic operations.
    \end{Theorem}

    Theorems~\ref{thm:product}, \ref{thm:distance}, 
    and~\ref{thm:matching} extend to piecewise-polynomial functions of
    constant maximum degree with essentially no modifications.
    What other classes of functions can be handles using similar methods?


    


    \subsection*{Acknowledgments}
    
    The authors would like to thank Raimund Seidel, Christian Knauer and
    Sylvain Lazard for helpful discussions. The authors were also inspired to
    work on this problem at the Workshop on Discrete and Algebraic Geometry
    in September 2010, at Val d'Ajol, France.

\end{document}